\documentclass[10pt]{article}%
\usepackage{amsmath}
\usepackage{amsfonts}
\usepackage{amssymb}
\usepackage{graphicx}
\usepackage[all]{xy}%
\newtheorem{theorem}{Theorem}

\newtheorem{corollary}[theorem]{Corollary}

\newtheorem{definition}[theorem]{Definition}
\newtheorem{example}[theorem]{Example}

\newtheorem{proposition}[theorem]{Proposition}
\newtheorem{remark}[theorem]{Remark}

\newenvironment{proof}[1][Proof]{\noindent\textbf{#1.} }{\ \rule{0.5em}{0.5em}}
\DeclareMathOperator{\sech}{sech}

\begin{document}

\title{Transmutations for Darboux transformed operators with applications}
\author{Vladislav V. Kravchenko and Sergii M. Torba\\{\small Department of Mathematics, CINVESTAV del IPN, Unidad Queretaro, }\\{\small Libramiento Norponiente No. 2000, Fracc. Real de Juriquilla,
Queretaro, }\\{\small Qro. C.P. 76230 MEXICO e-mail:
vkravchenko@qro.cinvestav.mx\thanks{Research was supported by CONACYT, Mexico.
Research of second named author was supported by DFFD, Ukraine (GP/F32/030)
and by SNSF, Switzerland (JRP IZ73Z0 of SCOPES 2009--2012). }}}
\maketitle

\begin{abstract}
We solve the following problem. Given a continuous complex-valued potential
$q_{1}$ defined on a segment $[-a,a]$ and let $q_{2}$ be the potential of a
Darboux transformed Schr\"{o}dinger operator, that is $q_{2}=-q_{1}%
+2\bigl(\frac{f^{\prime}}{f}\bigr)^{2}$ where $f$ is a nonvanishing solution
of the equation $A_{1}f = \big(\frac{d^{2}}{dx^{2}}-q_{1}(x)\big)f=0$. Suppose
a transmutation operator $\mathbf{T}_{1}$ is known such that $A_{1}%
\mathbf{T}_{1}u=\mathbf{T}_{1}\frac{d^{2}}{dx^{2}}u$ for any $u\in
C^{2}[-a,a]$. Find an analogous transmutation operator for $A_{2}=\frac{d^{2}%
}{dx^{2}}-q_{2}(x)$.

It is well known that the transmutation operators can be realized in the form
of Volterra integral operators with continuously differentiable kernels. Given
a kernel $K_{1}$ of the transmutation operator $\mathbf{T}_{1}$ we find the
kernel $K_{2}$ of $\mathbf{T}_{2}$ in a closed form in terms of $K_{1}$. As a
corollary interesting commutation relations between $\mathbf{T}_{1}$ and
$\mathbf{T}_{2}$ are obtained which then are used in order to construct the
transmutation operator for the one-dimensional Dirac system with a scalar potential.

\end{abstract}

\section{Introduction}

Transmutation operators are a widely used tool in the theory of linear
differential equations (see, e.g., \cite{Gilbert}, \cite{Carroll},
\cite{LevitanInverse}, \cite{Marchenko}, \cite{Trimeche} and the recent review
\cite{Sitnik}). It is well known that under certain quite general conditions
the transmutation operator transmuting the operator $A=-\frac{d^{2}}{dx^{2}%
}+q(x)$ into $B=-\frac{d^{2}}{dx^{2}}$ is a Volterra integral operator with
good properties. Its kernel can be obtained as a solution of a certain Goursat
problem for the Klein-Gordon equation with a variable coefficient. In the book
\cite{Fage} another approach to the transmutation was developed. It was shown
that to every (regular) linear second-order ordinary differential operator $L$
one can associate a linear space spanned on a so-called $L$-basis -- an
infinite family of functions $\left\{  \varphi_{k}\right\}  _{k=0}^{\infty}$
such that $L\varphi_{k}=0$ for $k=0,1$, $L\varphi_{k}=k(k-1)\varphi_{k-2}$,
for $k=2,3,\ldots$ and all $\varphi_{k}$ satisfy certain prescribed initial
conditions. Then the operator of transmutation was introduced as an operation
transforming functions from one such linear space corresponding to a certain
operator $L$ to functions from another linear space corresponding to another
operator $M$, and the transformation consists in substituting the $L$-basis
with the $M$-basis preserving the same coefficients in the expansion.

In \cite{CKT} the relation between the transmutation operator in the form of a
Volterra integral operator and as a transmutation of the bases was clarified.
In particular, it was shown that the corresponding Volterra integral
transmutation operator transmutes the powers of the independent variable
$x^{k}$ into the elements $\varphi_{k}$ of an appropriate $L$-basis. In order
to assure such an important mapping property of the transmutation operator a
parametrized family of transmutation operators for the Schr\"{o}dinger
operator $A$ was introduced. This parametrized family resulted to be a
necessary and natural tool for solving the problem of construction of the
transmutation operator for a Darboux transformed Schr\"{o}dinger operator. The
solution to this problem is the main aim of the present paper.

We give an explicit representation for the kernel of the transmutation
operator corresponding to the Darboux transformed potential in terms of the
transmutation kernel for its superpartner (Theorem \ref{Th T2Volterra}).
Moreover, this result leads to interesting commutation relations between the
two transmutation operators (Corollary \ref{Cor Commutation Relations}) which
in their turn allow us to obtain a transmutation operator for the
one-dimensional Dirac system with a scalar potential as well as to prove
(Theorem \ref{Th Transmutation C}) the main property of the transmutation
operator under less restrictive conditions than it has been proved until now.
Namely, we show that the considered transmutation operators transmute $A$ into
$B$ for any continuous (and not necessarily differentiable) potential $q $
admitting a nonvanishing, in general, complex-valued solution (and this
condition is always fulfilled, e.g., when $q$ is real valued). We give several
examples of explicitly constructed kernels of transmutation operators. It is
worth mentioning that in the literature there are very few explicit examples
and even in the case when $q$ is a constant such kernel was presented recently
in \cite{CKT}. The results of the present paper allow us to enlarge
considerably the list of available examples and give a relatively simple tool
for constructing Darboux related sequences of the transmutation kernels.

In Section \ref{Sect Transmutations and Recursive} we introduce two main
objects, the transmutation operators and the systems of recursive integrals or
the $L$-bases. Besides recalling some recent results from \cite{CKT} on the
parametrized family of the transmutation operators we prove (Theorem
\ref{TmainGoursat}) a criterion for a function to be a kernel of a
transmutation operator from the parametrized family as well as a result on the
explicit form of the inverse transmutation operator (Theorem \ref{Th Inverse}%
). In Section \ref{Sect Transmutations for Darboux} we present the
construction of the transmutations for the Darboux transformed operators and
obtain the commutation relations for them. We use the relations to weaken the
conditions on the potential under which the transmutation operators preserve
their main property (Theorem \ref{Th Transmutation C}). The section ends with
several examples of explicitly constructed transmutation kernels. In Section
\ref{Sect Transmutation Dirac} the results of the preceding section are used
in order to obtain the transmutation operator for the one-dimensional Dirac
system with a scalar potential. Finally, Section \ref{Sect Conclusions}
contains some concluding remarks.

\section{Transmutation operators for Sturm-Liouville equations and systems of
recursive integrals\label{Sect Transmutations and Recursive}}

\subsection{Transmutation operators}

\label{SubSectTransmSL} According to the definition given by Levitan
\cite{LevitanInverse}, let $E$ be a linear topological space, $A$ and $B$ be
linear operators: $E\rightarrow E$. Let $E_{1}$ and $E_{2}$ be closed
subspaces of $E$.

\begin{definition}
\label{DefTransmut} A linear invertible operator $T$ defined on the whole $E$
and acting from $E_{1}$ to $E_{2}$ is called a transmutation operator for the
pair of operators $A$ and $B$ if it fulfills the following two conditions.
\end{definition}

1. \textit{Both the operator }$T$\textit{\ and its inverse }$T^{-1}%
$\textit{\ are continuous in }$E$;

2. \textit{The following operator equality is valid}%

\begin{equation}
AT=TB \label{ATTB}%
\end{equation}

\textit{or which is the same}%
\[
A=TBT^{-1}.
\]

Our main interest concerns the situation when $A=-\frac{d^{2}}{dx^{2}}+q(x)$,
$B=-\frac{d^{2}}{dx^{2}}$, \ and $q$ is a continuous complex-valued function.
Hence for our purposes it will be sufficient to consider the functional space
$E=C^{2}[a,b]$ with the topology of uniform convergence. For simplicity we
will assume that the interval is symmetric with respect to the origin, thus
$E=C^{2}[-a,a]$.

An operator of transmutation for such $A$ and $B$ can be realized in the form
(see, e.g., \cite{LevitanInverse} and \cite{Marchenko}) of a Volterra integral
operator
\begin{equation}
Tu(x)=u(x)+\int_{-x}^{x}K(x,t)u(t)dt \label{T}%
\end{equation}
where $K(x,t)=H\big(\frac{x+t}{2},\frac{x-t}{2}\big)$ and $H$ is the unique
solution of the Goursat problem%
\begin{equation}
\frac{\partial^{2}H(u,v)}{\partial u\,\partial v}=q(u+v)H(u,v),
\label{GoursatH1}%
\end{equation}%
\begin{equation}
H(u,0)=\frac{1}{2}\int_{0}^{u}q(s)\,ds,\qquad H(0,v)=0. \label{GoursatH2}%
\end{equation}
If the potential $q$ is continuously differentiable, the kernel $K$ itself is
the solution of the Goursat problem
\begin{equation}
\left(  \frac{\partial^{2}}{\partial x^{2}}-q(x)\right)  K(x,t)=\frac
{\partial^{2}}{\partial t^{2}}K(x,t), \label{Goursat1}%
\end{equation}%
\begin{equation}
K(x,x)=\frac{1}{2}\int_{0}^{x}q(s)\,ds,\qquad K(x,-x)=0. \label{Goursat2}%
\end{equation}
If the potential $q$ is $n$ times continuously differentiable, the kernel
$K(x,t)$ is $n+1$ times continuously differentiable with respect to both
independent variables (see \cite{Marchenko}).

An important property of this transmutation operator consists in the way how
it maps solutions of the equation%
\begin{equation}
v^{\prime\prime}+\omega^{2}v=0 \label{SLomega1}%
\end{equation}
into solutions of the equation%
\begin{equation}
u^{\prime\prime}-q(x)u+\omega^{2}u=0 \label{SLomega2}%
\end{equation}
where $\omega$ is a complex number. Denote by $e_{0}(i\omega,x)$ the solution
of (\ref{SLomega2}) satisfying the initial conditions%
\begin{equation}
e_{0}(i\omega,0)=1\qquad\text{and}\qquad e_{0}^{\prime}(i\omega,0)=i\omega.
\label{initcond}%
\end{equation}
The subindex ``$0$'' indicates that the initial conditions correspond to the
point $x=0$ and the letter ``$e$'' reminds us that the initial values coincide
with the initial values of the function $e^{i\omega x}$.

The transmutation operator (\ref{T}) maps $e^{i\omega x}$ into $e_{0}%
(i\omega,x)$,
\begin{equation}
e_{0}(i\omega,x)=T[e^{i\omega x}] \label{e0=Te}%
\end{equation}
(see \cite[Theorem 1.2.1]{Marchenko}).

Following \cite{Marchenko} we introduce the following notations%
\[
K_{c}(x,t;h)=h+K(x,t)+K(x,-t)+h\int_{t}^{x}\{K(x,\xi)-K(x,-\xi)\}d\xi
\]
where $h$ is a complex number, and
\[
K_{s}(x,t;\infty)=K(x,t)-K(x,-t).
\]

\begin{theorem}
[\cite{Marchenko}]\label{TcTsMapsSolutions} Solutions $c(\omega,x;h)$ and
$s(\omega,x;\infty)$ of equation \eqref{SLomega2} satisfying the initial
conditions
\begin{align}
c(\omega,0;h)  &  =1, & c_{x}^{\prime}(\omega,0;h)  &  =h\label{ICcos}\\
s(\omega,0;\infty)  &  =0, & s_{x}^{\prime}(\omega,0;\infty)  &  =1
\label{ICsin}%
\end{align}
can be represented in the form
\begin{equation}
c(\omega,x;h)=\cos\omega x+\int_{0}^{x}K_{c}(x,t;h)\cos\omega t\,dt
\label{c cos}%
\end{equation}
and
\begin{equation}
s(\omega,x;\infty)=\frac{\sin\omega x}{\omega}+\int_{0}^{x}K_{s}%
(x,t;\infty)\frac{\sin\omega t}{\omega}\,dt. \label{s sin}%
\end{equation}

\end{theorem}

Denote by
\[
T_{c}u(x)=u(x)+\int_{0}^{x}K_{c}(x,t;h)u(t)dt
\]
and%
\[
T_{s}u(x)=u(x)+\int_{0}^{x}K_{s}(x,t;\infty)u(t)dt
\]
the corresponding integral operators. As was pointed out in \cite{CKT} they
are not, in general, transmutations on the whole space $C^{2}[-a,a]$.

\subsection{A parametrized family of transmutation operators}

In \cite{CKT} we introduced the parametrized family of operators
$\mathbf{T}_{h}$, $h\in\mathbb{C}$, given by the integral expression
\begin{equation}
\mathbf{T}_{h}u(x)=u(x)+\int_{-x}^{x}\mathbf{K}(x,t;h)u(t)dt \label{Tmain}%
\end{equation}
where
\begin{equation}
\mathbf{K}(x,t;h)=\frac{h}{2}+K(x,t)+\frac{h}{2}\int_{t}^{x}%
\big( K(x,s)-K(x,-s)\big) \,ds. \label{Kmain}%
\end{equation}
They are related to operators $T_{s}$ and $T_{c}$ (with the parameter $h$ in
the kernel of the latter operator) by
\begin{equation}
\mathbf{T}_{h}=T_{c}P_{e}+T_{s}P_{o}, \label{TcPe+TsPo}%
\end{equation}
where $P_{e}f(x)=\big(f(x)+f(-x)\big)/2$ and $P_{o}%
f(x)=\big(f(x)-f(-x)\big)/2$ are projectors to even and odd functions,
respectively. In this subsection we show that the operators $\mathbf{T}_{h}$
are transmutations, summarize their properties and later, in Theorem
\ref{Th Transmutation of Powers} we show how they act on powers of $x$.

Let us notice that $\mathbf{K}(x,t;0)=K(x,t)$ and that the expression
\[
\mathbf{K}(x,t;h)-\mathbf{K}(x,-t;h)=K(x,t)-K(x,-t)+\frac{h}{2}\int_{-t}%
^{t}\left(  K(x,s)-K(x,-s)\right)  \,ds=K(x,t)-K(x,-t)
\]
does not depend on $h$. Thus, it is possible to compute $\mathbf{K}(x,t;h)$
for any $h$ by a given $\mathbf{K}(x,t;h_{1})$ for some particular value
$h_{1}$. We formulate this result as the following statement.

\begin{theorem}
[\cite{CKT}]\label{Kh_and_Kh1} The integral kernels $\mathbf{K}(x,t;h)$ and
$\mathbf{K}(x,t;h_{1})$ are related by the expression
\begin{equation}
\mathbf{K}(x,t;h)=\frac{h-h_{1}}{2}+\mathbf{K}(x,t;h_{1})+\frac{h-h_{1}}%
{2}\int_{t}^{x}\big( \mathbf{K}(x,s;h_{1})-\mathbf{K}(x,-s;h_{1})\big) \,ds.
\label{KmainChangeOfH}%
\end{equation}

\end{theorem}

Using \eqref{e0=Te} let us check how the operators $\mathbf{T}_{h}$ act on
solutions of (\ref{SLomega1}).

\begin{proposition}
\label{ThMapsSolutions} The operator $\mathbf{T}_{h}$ maps a solution $v$ of
an equation $v^{\prime\prime}+\omega^{2}v=0$, where $\omega$ is a complex
number, into the solution $u$ of the equation $u^{\prime\prime}-q(x)u+\omega
^{2}u=0$ with the following correspondence of initial values
\[
u(0)=v(0),\qquad u^{\prime}(0)=v^{\prime}(0)+hv(0).
\]

\end{proposition}

\begin{proof}
It follows from \eqref{TcPe+TsPo} and Theorem \ref{TcTsMapsSolutions} that the
operator $\mathbf{T}_{h}$ indeed maps a solution of $v^{\prime\prime}%
+\omega^{2}v=0$ into a solution of $u^{\prime\prime}-q(x)u+\omega^{2}u=0$. It
is clear from the definition \eqref{Tmain} that
\begin{equation}
\label{Thu0}\mathbf{T}_{h}v(0)=v(0).
\end{equation}
As for the derivative, we have
\[
\big(\mathbf{T}_{h}v\big)^{\prime}(x)=v^{\prime}(x)+\int_{-x}^{x}
\mathbf{K}^{\prime}_{x}(x,t;h)v(t)\,dt+\mathbf{K}(x,x;h)v(x)+\mathbf{K}%
(x,-x;h)v(-x),
\]
therefore
\begin{equation}
\label{Thu'0}\big(\mathbf{T}_{h}v\big)^{\prime}(0)=v^{\prime}%
(0)+K(0,0;h)u(0)+K(x,-x;h)v(0)=v^{\prime}(0)+hv(0).
\end{equation}
\end{proof}

\begin{remark}
As can be seen from the proof, formulas \eqref{Thu0} and \eqref{Thu'0} are
valid for any function $v\in C^{1}[-a,a]$.
\end{remark}

We know that the kernel of the transmutation operator $T$ is related to the
solution of the Goursat problem \eqref{GoursatH1}--\eqref{GoursatH2}. We show
that a similar result holds for the operators $\mathbf{T}_{h}$.

\begin{theorem}
\label{TmainGoursat} In order for the function $K(x,t;h)$ to be the kernel of
a transmutation operator acting as described in Proposition
\ref{ThMapsSolutions}, it is necessary and sufficient that
$H(u,v;h):=K(u+v,u-v;h)$ be a solution of the Goursat problem
\begin{equation}
\frac{\partial^{2}H(u,v;h)}{\partial u\,\partial v}=q(u+v)H(u,v;h),
\label{GoursatTh1}%
\end{equation}%
\begin{equation}
H(u,0;h)=\frac{h}{2}+\frac{1}{2}\int_{0}^{u}q(s)\,ds,\qquad H(0,v;h)=\frac
{h}{2}. \label{GoursatTh2}%
\end{equation}
If the potential $q$ is continuously differentiable, the function $K(x,t;h)$
itself should be the solution of the Goursat problem
\begin{equation}
\left(  \frac{\partial^{2}}{\partial x^{2}}-q(x)\right)  K(x,t;h)=\frac
{\partial^{2}}{\partial t^{2}}K(x,t;h), \label{GoursatTk1}%
\end{equation}%
\begin{equation}
K(x,x;h)=\frac{h}{2}+\frac{1}{2}\int_{0}^{x}q(s)\,ds,\qquad K(x,-x;h)=\frac
{h}{2}. \label{GoursatTk2}%
\end{equation}

\end{theorem}

\begin{proof}
If some transmutation operator $T_{h}$ acts as described in Proposition
\ref{ThMapsSolutions}, then on a common domain of definition it coincides with
the constructed operator $\mathbf{T}_{h}$, because both operators are
continuous and act identically on the dense subset of functions, namely, on
linear combinations of solutions of equations (\ref{SLomega1}). Hence it is
sufficient to verify that the kernel $\mathbf{K}(x,t;h)$ satisfies conditions
of the theorem.

Suppose $q\in C^{1}[-a,a]$. Then the kernel $K(x,t)$ is twice continuously
differentiable and we can check \eqref{GoursatTk1} directly. In the following
calculation $K_{x}^{\prime}(x,\pm x)$ means $K_{x}^{\prime}(x,t)\big|_{t=\pm
x}$ and $K_{t}^{\prime}(x,\pm x)$ means $K_{t}^{\prime}(x,t)\big|_{t=\pm x}$.
\begin{multline*}
\left(  \partial_{x}^{2}-\partial_{t}^{2}\right)  \mathbf{K}(x,t;h)=\left(
\partial_{x}^{2}-\partial_{t}^{2}\right)  \left[  \frac{h}{2}+K(x,t)+\frac
{h}{2}\int_{t}^{x}\left(  K(x,s)-K(x,-s)\right)  \,ds\right]  =\\
=q(x)K(x,t)+\frac{h}{2}\partial_{x}\left(  \int_{t}^{x}\big(K_{x}^{\prime
}(x,s)-K_{x}^{\prime}(x,-s)\big)\,ds+K(x,x)-K(x,-x)\right)  +\\
+\frac{h}{2}\partial_{t}\big(K(x,t)-K(x,-t)\big)=q(x)K(x,t)+\frac{h}{4}q(x)+\\
+\frac{h}{2}\bigg(\int_{t}^{x}\big(K_{xx}^{\prime\prime}(x,s)-K_{xx}%
^{\prime\prime}(x,-s)\big)\,ds+K_{x}^{\prime}(x,x)-K_{x}^{\prime}%
(x,-x)+K_{t}^{\prime}(x,t)+K_{t}^{\prime}(x,-t)\bigg)=\\
=q(x)K(x,t)+\frac{h}{4}q(x)+\frac{h}{2}\bigg(\int_{t}^{x}\big(K_{ss}%
^{\prime\prime}(x,s)-K_{ss}^{\prime\prime}%
(x,-s)+q(x)(K(x,s)-K(x,-s))\big)\,ds+\\
+K_{x}^{\prime}(x,x)-K_{x}^{\prime}(x,-x)+K_{t}^{\prime}(x,t)+K_{t}^{\prime
}(x,-t)\bigg)=q(x)\bigg(K(x,t)+\frac{h}{4}+\int_{t}^{x}%
\big(K(x,s)-K(x,-s)\big)\,ds\bigg)+\\
+\frac{h}{2}\big(K_{t}^{\prime}(x,x)-K_{t}^{\prime}(x,t)+K_{t}^{\prime
}(x,-x)-K_{t}^{\prime}(x,-t)+K_{x}^{\prime}(x,x)-K_{x}^{\prime}(x,-x)+K_{t}%
^{\prime}(x,t)+K_{t}^{\prime}(x,-t)\big)=\\
=q(x)\bigg(K(x,t)+\frac{h}{2}+\int_{t}^{x}%
\big(K(x,s)-K(x,-s)\big)\,ds\bigg)=q(x)\mathbf{K}(x,t;h),
\end{multline*}
where we have used \eqref{Goursat1}, \eqref{Goursat2} and the equalities
$K_{x}^{\prime}(x,x)+K_{t}^{\prime}(x,x)=\partial_{x}(K(x,x))=1/2q(x)$,
$K_{x}^{\prime}(x,-x)-K_{t}^{\prime}(x,-x)=\partial_{x}(K(x,-x))=0$.
Equalities \eqref{GoursatTk2} are valid by the definition \eqref{Kmain} of
$\mathbf{K}(x,t;h)$. Clearly, the function $\mathbf{H}(u,v;h)=\mathbf{K}%
(u+v,u-v;h)$ is also the solution of the Goursat problem \eqref{GoursatTh1}--\eqref{GoursatTh2}.

Suppose now that $q\in C[-a,a]$. Consider a sequence of functions
$\{q_{n}\}_{n\in\mathbb{N}}$, where $q_{n}\in C^{1}[-a,a]$, $n\in\mathbb{N}$
and $q_{n}\rightarrow q,\ n\rightarrow\infty$ uniformly on $[-a,a]$. Denote by
$H_{n}(u,v)$ the solutions of Goursat problems
\eqref{GoursatH1}--\eqref{GoursatH2} with the potentials $q_{n}$, and by
$H(u,v)$ the solution with the potential $q$. It is known (see, e.g.,
\cite{Marchenko}), that $H_{n}(u,v)\rightarrow H(u,v),\ n\rightarrow\infty$
uniformly. The same is valid for the functions $K_{n}(x,t)=H_{n}%
\big(\frac{x+t}{2},\frac{x-t}{2}\big)$ and $K(x,t)$. Denote by $\mathbf{K}%
_{n}(x,t;h)$ the integral kernels constructed from $K_{n}(x,t)$ by means of
\eqref{Kmain} and let $\mathbf{H}_{n}(u,v;h)=\mathbf{K}_{n}(u+v,u-v;h)$. Then
$\mathbf{H}_{n}(u,v;h)\rightarrow\mathbf{H}(u,v;h)=\mathbf{K}%
(u+v,u-v;h),\ n\rightarrow\infty$ uniformly. Similarly to \cite{Marchenko,
Vladimirov}, the function $\mathbf{H}_{n}(u,v;h)$ is the solution of the
Goursat problem \eqref{GoursatTh1}--\eqref{GoursatTh2} with the potential
$q_{n}$ if and only if $\mathbf{H}_{n}(u,v;h)$ satisfies the integral
equation
\begin{equation}
\mathbf{H}_{n}(u,v;h)=\frac{h}{2}+\frac{1}{2}\int_{0}^{u}q_{n}(y)\,dy+\int
_{0}^{u}d\alpha\,\int_{0}^{v}q_{n}(\alpha+\beta)\mathbf{H}_{n}(\alpha
,\beta;h)\,d\beta. \label{IntegralEqGoursat}%
\end{equation}
From the first part of the proof we know that the functions $\mathbf{H}%
_{n}(u,v;h)$ satisfy \eqref{IntegralEqGoursat}. Passing to the limit when
$n\rightarrow\infty$ in \eqref{IntegralEqGoursat}, we obtain that
$\mathbf{H}(u,v;h)$ satisfies \eqref{IntegralEqGoursat} with the potential $q$
and hence satisfies the Goursat problem \eqref{GoursatTh1}--\eqref{GoursatTh2}.

For the second part of the theorem, suppose that a function $H(u,v;h)$
satisfies the Goursat problem \eqref{GoursatTh1}--\eqref{GoursatTh2}. Consider
the Goursat problem \eqref{GoursatH1}--\eqref{GoursatH2}. It is well posed,
and from its solution using \eqref{Kmain} we can construct the kernel
$\mathbf{K}(x,t;h)$. From the first part of the proof we have that the
function $\mathbf{H}(u,v;h)=\mathbf{K}(u+v,u-v;h)$ satisfies the Goursat
problem \eqref{GoursatTh1}--\eqref{GoursatTh2}, and from the uniqueness of the
solution it follows that $H=\mathbf{H}$ and $H$ gives us the transmutation
operator acting as described in Proposition \ref{ThMapsSolutions}.
\end{proof}

\begin{example}
\label{modelExample} Consider a function $k(x,t)=\frac{t-1}{2(x+1)}$ (later,
in Subsection \ref{subsectExamples} it is explained how it can be obtained).
We have
\begin{gather*}
(\partial_{x}^{2}-\partial_{t}^{2})k(x,t)=\frac{t-1}{(x+1)^{3}}=\frac
{2}{(x+1)^{2}}\cdot\frac{t-1}{2(x+1)},\\
k(x,-x)=\frac{-x-1}{2(x+1)}=-\frac{1}{2}\qquad\text{and}\qquad k(x,x)=\frac
{x-1}{2(x+1)}=-\frac{1}{2}+\frac{1}{2}\int_{0}^{x}\frac{2}{(s+1)^{2}}\,ds,
\end{gather*}
thus the function $k(x,t)$ satisfies the Goursat problem
\eqref{GoursatTk1}--\eqref{GoursatTk2} with $q(x)=2/(x+1)^{2}$ and $h=-1$ and
by Theorem \ref{TmainGoursat} is the kernel of the transmutation operator
$\mathbf{T}_{-1}$.

Consider two solutions of the equation $u^{\prime\prime}=0$, $u_{1}=1$ and
$u_{2}=x$. We have that their images
\begin{align*}
v_{1}=\mathbf{T}_{-1}u_{1}  &  =1+\int_{-x}^{x}\frac{t-1}{2(x+1)}\,dt=\frac
{1}{x+1},\\
v_{2}=\mathbf{T}_{-1}u_{2}  &  =x+\int_{-x}^{x}\frac{(t-1)t}{2(x+1)}%
\,dt=\frac{x^{3}+3x^{2}+3x}{3(x+1)},
\end{align*}
are solutions of the equation $v^{\prime\prime}-2/(x+1)^{2}v=0$ with the
initial values in agreement with Proposition \ref{ThMapsSolutions}. For the
solution $u_{3}=\cos x$ of the equation $u^{\prime\prime}=-u$ we get that
\[
v_{3}=\mathbf{T}_{-1}u_{3}=\cos x+\int_{-x}^{x}\frac{(t-1)\cos t}%
{2(x+1)}\,dt=\cos x-\frac{\sin x}{x+1},
\]
is a solution of the equation $v^{\prime\prime}-\frac{2}{(x+1)^{2}}v=-v$.
\end{example}

For the case $q\in C^{1}[-a,a]$ the Volterra-type integral operator \eqref{T}
is a transmutation in the sense of Definition \ref{DefTransmut} on the space
$C^{2}[-a,a]$ if and only if the integral kernel $K(x,t)$ satisfies the
Goursat problem
\begin{gather*}
\left(  \frac{\partial^{2}}{\partial x^{2}}-q(x)\right)  K(x,t)=\frac
{\partial^{2}}{\partial t^{2}}K(x,t),\\
K(x,x)=C+\frac{1}{2}\int_{0}^{x}q(s)\,ds,\qquad K(x,-x)=C,
\end{gather*}
where $C$ is a constant (see \cite{Lions57, Trimeche}). The proof of this fact
consists in differentiating of the integral and integrating by parts, similar
to the proof of Theorem \ref{TmainGoursat}. Due to Theorem \ref{TmainGoursat}
we see that $\mathbf{T}_{h}$ is indeed a transmutation.

\begin{theorem}
\label{Th Transmutation} Let $q\in C^{1}[-a,a]$. Then the operator
$\mathbf{T}_{h}$ given by \eqref{Tmain} satisfies the equality
\begin{equation}
\left(  -\frac{d^{2}}{dx^{2}}+q(x)\right)  \mathbf{T}_{h}[u]=\mathbf{T}%
_{h}\left[  -\frac{d^{2}}{dx^{2}}(u)\right]  \label{ThTransm}%
\end{equation}
for any $u\in C^{2}[-a,a]$.
\end{theorem}

\begin{remark}
This theorem was proved in \cite{CKT} under the additional assumption that
there exists a nonvanishing solution $f$ of $f^{\prime\prime}-qf=0$ on
$[-a,a]$. Later, in Theorem \ref{Th Transmutation C} we show that such
additional assumption allows one to weaken the requirement on the potential
$q$ to $q\in C[-a,a]$ keeping valid equality (\ref{ThTransm}).
\end{remark}

As can be seen from \eqref{Tmain}, to define the transmutation operator we
have to know its integral kernel in the region $-a<x<a$, $|t|<|x|$. The
corresponding function $H$ should be defined in two domains, upper-right and
lower-left triangles $\Pi_{1}:\ u>0,v>0,u+v<a$ and $\Pi_{3}%
:\ u<0,v<0,|u|+|v|<a$ meanwhile the Goursat problem
\eqref{GoursatTh1}--\eqref{GoursatTh2} may be considered in four domains,
triangles $\Pi_{1}$ and $\Pi_{3}$ and triangles $\Pi_{2}:\ u<0,v>0,|u|+|v|<a$,
$\Pi_{4}:\ u>0,v<0,|u|+|v|<a$. The potential $q$ is continuous in each of
these domains, boundary conditions \eqref{GoursatTh2} are continuously
differentiable, therefore the corresponding solutions of the Goursat problems
(see \cite{Vladimirov}) belong to $C^{1}(\bar{\Pi}_{i}),\ i=1\ldots4 $. Hence
the function $\mathbf{H}(u,v;h)$ is defined in the combined domain
$|u|+|v|\leq a$ and belongs to the class $C^{1}$ on it. Consequently, the
function $\mathbf{K}(x,t;h)=\mathbf{H}\big(\frac{x+t}{2},\frac{x-t}%
{2};h\big)
$ is defined on the domain $\bar{\Pi}:\ -a\leq x\leq a,-a\leq t\leq a$ and is
continuously differentiable there. For the rest of this paper we assume that
the kernel $\mathbf{K}$ is defined on this larger domain $\bar{\Pi}$.

Let us find the inverse operator $\mathbf{T}_{h}^{-1}$. Since $\mathbf{T}_{h}
$ is the Volterra integral operator its inverse (see \cite{Lions57}) is again
a Volterra integral operator
\[
\mathbf{T}_{h}^{-1}u(x)=u(x)+\int_{-x}^{x}L(x,t;h)u(t)\,dt,
\]
where $L(x,t;h)$ satisfies (for $q\in C^{1}[-a,a]$) the Goursat problem
\begin{equation}
\frac{\partial^{2}}{\partial x^{2}}L(x,t;h)-\left(  \frac{\partial^{2}%
}{\partial t^{2}}-q(t)\right)  L(x,t;h)=0, \label{GoursatL1}%
\end{equation}%
\begin{equation}
L(x,x;h)=C_{1}-\frac{1}{2}\int_{0}^{x}q(s)\,ds,\qquad L(x,-x;h)=C_{1}.
\label{GoursatL2}%
\end{equation}
To determine the value of the constant $C_{1}$ we use Proposition
\ref{ThMapsSolutions}. Namely, we observe that the integral operator with the
kernel $L$ satisfying the Goursat problem \eqref{GoursatL1}--\eqref{GoursatL2}
gives us the following correspondence of initial values
\[
v(0)=u(0),\qquad v^{\prime}(0)=u^{\prime}(0)+2C_{1}u(0),
\]
where $v=\mathbf{T}_{h}^{-1}u$. Therefore for the operator to be the inverse
of the operator $\mathbf{T}_{h}$ we should take $C_{1}=-h/2$. Comparing the
Goursat problem \eqref{GoursatL1}--\eqref{GoursatL2} with $C_{1}=-h/2$ with
\eqref{GoursatTk1}--\eqref{GoursatTk2} we conclude that $L(x,t;h):=-\mathbf{K}%
(t,x;h)$ is the unique solution of \eqref{GoursatL1}--\eqref{GoursatL2} and
hence it is the kernel of the inverse operator. The assumption $q\in
C^{1}[-a,a]$ is not essential and can be easily overcome by considering an
approximating sequence of continuously differentiable potentials $q_{n}$.
Thus, the following statement is proved.

\begin{theorem}
\label{Th Inverse}The inverse operator $\mathbf{T}_{h}^{-1}$ can be
represented as the Volterra integral operator
\begin{equation}
\mathbf{T}_{h}^{-1}u(x)=u(x)-\int_{-x}^{x}\mathbf{K}(t,x;h)u(t)\,dt.
\label{Tinverse}%
\end{equation}

\end{theorem}

\begin{example}
For the operator $\mathbf{T}_{-1}$ with the kernel $K(x,t)$ and functions
$u_{i}$ and $v_{i}$, $i=1,2,3$ from Example \ref{modelExample} it is easy to
see that indeed
\[
u_{i}(x)=\mathbf{T}_{-1}^{-1}[v_{i}](x)=v_{i}(x)-\int_{-x}^{x}\frac
{x-1}{2(t+1)}v_{i}(t)\,dt,\qquad i=1,2,3.
\]

\end{example}

\subsection{A complete system of recursive integrals}

Let $f\in C^{2}(a,b)\cap C^{1}[a,b]$ be a complex valued function and
$f(x)\neq0$ for any $x\in[a,b]$. The interval $(a,b)$ is supposed to be
finite. Let us consider the following auxiliary functions%
\begin{align}
\widetilde{X}^{(0)}(x)  &  \equiv X^{(0)}(x)\equiv1,\label{X1}\\
\widetilde{X}^{(n)}(x)  &  =n\int_{x_{0}}^{x}\widetilde{X}^{(n-1)}(s)\left(
f^{2}(s)\right)  ^{(-1)^{n-1}}\,\mathrm{d}s,\label{X2}\\
X^{(n)}(x)  &  =n\int_{x_{0}}^{x}X^{(n-1)}(s)\left(  f^{2}(s)\right)
^{(-1)^{n}}\,\mathrm{d}s, \label{X3}%
\end{align}
where $x_{0}$ is an arbitrary fixed point in $[a,b]$. We introduce the
infinite system of functions $\left\{  \varphi_{k}\right\}  _{k=0}^{\infty}$
defined as follows
\begin{equation}
\varphi_{k}(x)=
\begin{cases}
f(x)X^{(k)}(x), & k \text{\ odd,}\\
f(x)\widetilde{X}^{(k)}(x), & k \text{\ even,}%
\end{cases}
\label{phik}%
\end{equation}
where the definition of $X^{(k)}$ and $\widetilde{X}^{(k)}$ is given by
(\ref{X1})--(\ref{X3}) with $x_{0}$ being an arbitrary point of the interval
$[a,b]$.

\begin{example}
\label{ExamplePoly}Let $f\equiv1$, $a=0$, $b=1$. Then it is easy to see that
choosing $x_{0}=0$ we have $\varphi_{k}(x)=x^{k}$, $k\in\mathbb{N}_{0}$ where
by $\mathbb{N}_{0}$ we denote the set of non-negative integers.
\end{example}

In \cite{KrCMA2011} it was shown that the system $\left\{  \varphi
_{k}\right\}  _{k=0}^{\infty}$ is complete in $L_{2}(a,b)$ and in \cite{KMoT}
its completeness in the space of piecewise differentiable functions with
respect to the maximum norm was obtained and series expansions in terms of the
functions $\varphi_{k}$ were studied.

The system (\ref{phik}) is closely related to the notion of the $L$-basis
introduced and studied in \cite{Fage}. Here the letter $L$ corresponds to a
linear ordinary differential operator. This becomes more transparent from the
following result obtained in \cite{KrCV08} (for additional details and simpler
proof see \cite{APFT} and \cite{KrPorter2010}) establishing the relation of
the system of functions $\left\{  \varphi_{k}\right\}  _{k=0}^{\infty}$ to
Sturm-Liouville equations.

\begin{theorem}
[\cite{KrCV08}]\label{ThGenSolSturmLiouville} Let $q$ be a continuous complex
valued function of an independent real variable $x\in\lbrack a,b],$ $\lambda$
be an arbitrary complex number. Suppose there exists a solution $f$ of the
equation
\begin{equation}
f^{\prime\prime}-qf=0 \label{SLhom}%
\end{equation}
on $(a,b)$ such that $f\in C^{2}(a,b)$ together with $1/f$ are bounded on
$[a,b]$. Then the general solution of the equation
\begin{equation}
u^{\prime\prime}-qu=\lambda u \label{SLlambda}%
\end{equation}
on $(a,b)$ has the form%
\[
u=c_{1}u_{1}+c_{2}u_{2}%
\]
where $c_{1}$ and $c_{2}$ are arbitrary complex constants,
\begin{equation}
u_{1}=\sum_{k=0}^{\infty}\frac{\lambda^{k}}{(2k)!}\varphi_{2k}\quad
\text{and}\quad u_{2}= \sum_{k=0}^{\infty} \frac{\lambda^{k}}{(2k+1)!}%
\varphi_{2k+1} \label{u1u2}%
\end{equation}
and both series converge uniformly on $[a,b]$.
\end{theorem}

\begin{remark}
\label{RemInitialValues}It is easy to see that by definition the solutions
$u_{1}$ and $u_{2}$ satisfy the following initial conditions
\begin{equation}
u_{1}(x_{0})=f(x_{0}),\qquad u_{1}^{\prime}(x_{0})=f^{\prime}(x_{0}),
\label{initial1}%
\end{equation}%
\begin{equation}
u_{2}(x_{0})=0,\qquad u_{2}^{\prime}(x_{0})=1/f(x_{0}). \label{initial2}%
\end{equation}

\end{remark}

\subsection{Transmutations and systems of recursive integrals}

In this subsection we show the connection between the transmutation operators
$\mathbf{T}_{h}$ and the functions $\varphi_{k}$. All the results of this
subsection were proved in \cite{CKT}. Here we recall them with explanations.

We suppose that $f$ is a solution of (\ref{SLhom}) fulfilling the condition of
Theorem \ref{ThGenSolSturmLiouville} on a finite interval $(-a,a)$. We
normalize $f$ in such a way that $f(0)=1$, and let $f^{\prime}(0)=h$ where $h
$ is some complex constant. Let us obtain the expansion of the solution
$c(\omega,x;h)$ from Subsection \ref{SubSectTransmSL} in terms of the
functions $\varphi_{k}$. According to Remark \ref{RemInitialValues} the
solutions (\ref{u1u2}) of equation (\ref{SLlambda}) have the following initial
values%
\[
u_{1}(0)=1,\qquad u_{1}^{\prime}(0)=h,\qquad u_{2}(0)=0,\qquad u_{2}^{\prime
}(0)=1.
\]
Hence due to \eqref{ICcos} we obtain $c(\omega,x;h)=u_{1}(x)$. From \eqref{c
cos} and \eqref{u1u2} we have the equality%
\[
\sum_{k=0}^{\infty}\frac{(i\omega)^{2k}}{(2k)!}\varphi_{2k}(x)=\sum
_{j=0}^{\infty}\frac{(i\omega)^{2j}x^{2j}}{(2j)!}+\int_{0}^{x}\bigg(K_{c}%
(x,t;h)\sum_{j=0}^{\infty}\frac{(i\omega)^{2j}t^{2j}}{(2j)!}\bigg)dt.
\]
As the series under the sign of integral converges uniformly and the kernel
$K_{c}(x,t;h)$ is at least continuously differentiable (for a continuous $q$
\cite{Marchenko}) we obtain the following relation%
\[
\sum_{k=0}^{\infty}\frac{(i\omega)^{2k}}{(2k)!}\varphi_{2k}(x)=\sum
_{j=0}^{\infty}\frac{(i\omega)^{2j}}{(2j)!}\bigg(x^{2j}+\int_{0}^{x}%
K_{c}(x,t;h)\,t^{2j}dt\bigg).
\]
This equality holds for any $\omega$ hence we obtain the termwise relations%
\begin{equation}
\varphi_{2k}=T_{c}[x^{2k}],\quad k\in\mathbb{N}_{0}. \label{Tc xk}%
\end{equation}
Similarly, due to \eqref{ICsin} we observe that $s(\omega,x;\infty)=u_{2}(x)$
and from \eqref{s sin} we find that
\begin{equation}
\varphi_{2k+1}=T_{s}[x^{2k+1}],\quad k\in\mathbb{N}_{0}. \label{Ts xk}%
\end{equation}
Hence from the last two equalities and \eqref{TcPe+TsPo} we conclude that the
following statement is true.

\begin{theorem}
[\cite{CKT}]\label{Th Transmutation of Powers} Let $q$ be a continuous complex
valued function of an independent real variable $x\in\lbrack-a,a]$, and $f$ be
a particular solution of (\ref{SLhom}) such that $f\in C^{2}(-a,a)$ together
with $1/f$ are bounded on $[-a,a]$ and normalized as $f(0)=1$, $f^{\prime
}(0)=h$, where $h$ is a complex number. Then the operator (\ref{Tmain}) with
the kernel defined by (\ref{Kmain}) transforms $x^{k}$ into $\varphi_{k}(x)$
for any $k\in\mathbb{N}_{0}$.
\end{theorem}

Thus, we clarified what is the result of application of the transmutation
$\mathbf{T}_{h}$ to the powers of the independent variable. This is very
useful due to the fact that as a rule the construction of the kernel
$\mathbf{K}(x,t;h)$ in a more or less explicit form up to now is impossible.
Our result offers an algorithm for transmuting functions which can be
represented or at least approximated by finite or infinite polynomials in the
situation when $\mathbf{K}(x,t;h)$ is unknown.

\begin{remark}
Let $f$ be the solution of (\ref{SLhom}) satisfying the initial conditions
\begin{equation}
f(0)=1\quad\text{and}\quad f^{\prime}(0)=0. \label{initcond 1 0}%
\end{equation}
If it does not vanish on $[-a,a]$ then from Theorem
\ref{Th Transmutation of Powers} we obtain that the original transmutation
operator $T$ transmutes powers of the independent variable into $\varphi_{k}$:
$\varphi_{k}=T[x^{k}]$ for any $k\in\mathbb{N}_{0}$. In general, of course $f$
may have zeros on $[-a,a]$ and hence one can guarantee that the operator $T$
transmutes the powers of $x$ into $\varphi_{k}$ whose construction is based on
the solution $f$ satisfying (\ref{initcond 1 0}) only in some neighborhood of
the origin. If we construct $\varphi_{k}$ by a solution $f$ such that $f(0)=1$
and $f^{\prime}(0)=h$, then the original transmutation operator $T$ does not
map powers of the independent variable exactly to the functions $\varphi_{k}$.
As it was shown in \cite{CKT} the following equalities are valid
\begin{align*}
\varphi_{k}  &  =T[x^{k}]\text{\quad when }k\text{ is odd,}\\
\varphi_{k}-\frac{h}{k+1}\varphi_{k+1}  &  =T[x^{k}]\text{\quad when }%
k\in\mathbb{N}_{0}\text{ is even.}%
\end{align*}
This observation explains the necessity to consider the parametrized family of
transmutation operators whose mapping properties are well adjusted to
corresponding families of the functions $\varphi_{k}$.
\end{remark}

\begin{example}
For the operator $\mathbf{T}_{-1}$ from Example \ref{modelExample}, consider
the function $f=\mathbf{T}_{-1}[1]=\frac{1}{x+1}$ as a solution of
\eqref{SLhom} such that $f(0)=1$ and $f^{\prime}(0)=h=-1$, nonvanishing on any
$[-a,a]\subset(-1,1)$. The first 3 functions $\varphi_{k}$ are given by
\[
\varphi_{0}=f=\frac{1}{x+1},\quad\varphi_{1}=\frac{x^{3}+3x^{2}+3x}%
{3(x+1)},\quad\varphi_{2}=\frac{2x^{3}+3x^{2}}{3(x+1)}.
\]
We have seen already that $\varphi_{0}=\mathbf{T}_{-1}[1]$ and $\varphi
_{1}=\mathbf{T}_{-1}[x]$. We have also $\mathbf{T}_{-1}[x^{2}]=x^{2}%
-\frac{x^{3}}{3(x+1)}=\varphi_{2}$.
\end{example}

\section{Transmutation operators for Darboux transformed
equations\label{Sect Transmutations for Darboux}}

\subsection{The Darboux transformation}

Consider a Sturm-Liouville operator $A_{1}:=\frac{d^{2}}{dx^{2}}-q_{1}(x)$,
where $q_{1}$ is a continuous complex-valued function on some segment
$[-a,a]$. Suppose a solution $f$ of the equation $A_{1}f=0$ is given such that
$f(x)\neq0,\ x\in\lbrack-a,a]$, it is normalized as $f(0)=1$ and
$h:=f^{\prime}(0)$ is some complex number. Then the operator $A_{2}%
:=\frac{d^{2}}{dx^{2}}-q_{2}(x)$, where $q_{2}(x)=2\big(\frac{f^{\prime}%
(x)}{f(x)}\big)^{2}-q_{1}(x)$, is known as the Darboux transformation of the
operator $A_{1}$.

Initially the Darboux transformation served for establishing a relation
between the general solution $u$ of the equation $A_{1}u=\lambda u$ and the
general solution $v$ of the equation $A_{2}v=\lambda v$ by a known particular
solution $f$ of the equation $A_{1}u=0$. This relation is given by the formula
$v(x)=u^{\prime}(x)-u(x)\frac{f^{\prime}(x)}{f(x)}$. Later on it was found
that the Darboux transformation was closely related to the factorization of
the Schr\"{o}dinger equation, and nowadays it is used in dozens of works, see
e.g. \cite{Cies, GHZh, Matveev, RS} in connection with solitons and integrable
systems, e.g. \cite{BS, HV, NPS, PPS} and the review \cite{Rosu} of
applications to quantum mechanics.

We remind some well known facts about the Darboux transformation. First, $1/f
$ is the solution of $A_{2}u=0$. Second, it is closely related to the
factorization of Sturm-Liouville and one-dimensional Schr\"{o}dinger
operators. Namely, we have
\begin{align}
A_{1}=\frac{d^{2}}{dx^{2}}-q_{1}(x)  &  =\Big(\partial_{x}+\frac{f^{\prime}%
}{f}\Big)\Big(\partial_{x}-\frac{f^{\prime}}{f}\Big)=\frac{1}{f}\partial
_{x}f^{2}\partial_{x}\frac{1}{f}\cdot,\label{A1factor}\\
A_{2}=\frac{d^{2}}{dx^{2}}-q_{2}(x)  &  =\Big(\partial_{x}-\frac{f^{\prime}%
}{f}\Big)\Big(\partial_{x}+\frac{f^{\prime}}{f}\Big)=f\partial_{x}\frac
{1}{f^{2}}\partial_{x}f\cdot. \label{A2factor}%
\end{align}
Suppose that $u$ is a solution of the equation $A_{1}u=\omega u$ for some
$\omega\in\mathbb{C}$. Then the function $v=\big(\partial_{x}-\frac{f^{\prime
}}{f}\big)u=\big(f\partial_{x}\frac{1}{f}\big)u$ is a solution of the equation
$A_{2}v=\omega v$, and vice versa, given a solution $v$ of $A_{2}v=\omega v$,
the function $u=\big(\partial_{x}+\frac{f^{\prime}}{f}\big)v=\big(\frac{1}%
{f}\partial_{x}f\big)v$ is a solution of $A_{1}u=\omega u$.

\subsection{Construction of transmutation operators for Darboux transformed
equations}

Consider the operators $A_{1}$, $A_{2}$ and a particular solution $f$ of the
equation $A_{1}f=0$ as in the previous subsection. Suppose that the operator
$\mathbf{T}_{1;h}$ which transmutes the operator $A_{1}$ into the operator
$B=d^{2}/dx^{2}$ is known in the sense that its kernel $\mathbf{K}_{1}(x,t;h)$
is given. As before $h=f^{\prime}(0)$ and $\mathbf{T}_{1;h}$ transforms
solutions according to Proposition \ref{ThMapsSolutions}. As the parameter $h
$ is fixed by the value $f^{\prime}(0)$ we will simply write $\mathbf{T}_{1}$
instead of $\mathbf{T}_{1;h}$. We know that the function $1/f$ is the
non-vanishing solution of the equation $A_{2}u=0$ satisfying $1/f(0)=1$ and
$(1/f)^{\prime}(0)=-h$. Hence it is natural to look for the operator
$\mathbf{T}_{2;-h}$ transmuting the operator $A_{2}$ into the operator $B$. We
will simply write $\mathbf{T}_{2}$ further in this paper.

Let us explain the idea for obtaining the operator $\mathbf{T}_{2}$. We want
to find an operator transforming solutions of the equation $Bu+\omega^{2}u=0$
into solutions of the equation $A_{2}u+\omega^{2}u=0$, see the first diagram
below. Starting with a solution $\sigma$ of the equation $(\partial_{x}%
^{2}+\omega^{2})\sigma=0$, by application of $\mathbf{T}_{1}$ we get a
solution of $(A_{1}+\omega^{2})u=0$, and the expression $\big(f\partial
_{x}\frac{1}{f}\big)\mathbf{T}_{1}\sigma$ is a solution of $(A_{2}+\omega
^{2})v=0$. But the operator $\big(f\partial_{x}\frac{1}{f}\big)\mathbf{T}_{1}$
is unbounded and hence cannot coincide with the operator $\mathbf{T}_{2}$. In
order to find the required bounded operator we may consider the second copy of
the equation $(\partial_{x}^{2}+\omega^{2})u=0$, which is a result of the
Darboux transformation applied to $(\partial_{x}^{2}+\omega^{2})\sigma=0$ with
respect to the particular solution $g\equiv1$ and construct the operator
$\mathbf{T}_{2}$ by making the second diagram commutative. In order to obtain
a bounded operator $\mathbf{T}_{2}$, instead of using $f\partial_{x}\frac
{1}{f}$ for the last step, we will use the inverse of $\frac{1}{f}\partial
_{x}f$, i.e. $\frac{1}{f}\big(\int_{0}^{x}f(s)\cdot\,ds+C\big)$.
\[
\xymatrix@R+1pt@C+12pt@M+1pt{
\partial_{x}^{2}+\omega^{2} \ar[r]^(.45){\mathbf{T}_1} \ar@/_/[dr]_(.45){\mathbf{T}_2} & \partial_{x}^{2}-q_{1}+\omega^{2} \ar[d]^{f\partial_{x}\frac{1}{f}}\\
& \partial_{x}^{2}-q_{2}+\omega^{2}
}\qquad\qquad\qquad
\xymatrix@R+1pt@C+12pt@M+1pt{
\partial_{x}^{2}+\omega^{2} \ar[r]^(.45){\mathbf{T}_1} & \partial_{x}^{2}-q_{1}+\omega^{2} \ar@<1ex>[d]^{\frac{1}{f}(\int
f\cdot+C)}\\
\partial_{x}^{2}+\omega^{2} \ar[u]^{\partial_{x}} \ar[r]^(.45){\mathbf{T}_2} & \partial_{x}^{2}-q_{2}+\omega^{2} \ar@<1ex>[u]^{\frac{1}{f}
\partial_{x}f}
}
\]

That explains how to obtain the following theorem.

\begin{theorem}
\label{Th T2Integral} The operator $T_{2}$, acting on solutions $u$ of
equations $(\partial_{x}^{2}+\omega^{2})u=0,\ \omega\in\mathbb{C}$ by the
rule
\begin{equation}
T_{2}[u](x)=\frac{1}{f(x)}\bigg(\int_{0}^{x}f(\eta)\mathbf{T}_{1}[u^{\prime
}](\eta)\,d\eta+u(0)\bigg) \label{T2sol}%
\end{equation}
coincides with the transmutation operator $\mathbf{T}_{2;-h}$.
\end{theorem}

\begin{proof}
According to Proposition \ref{ThMapsSolutions} the operator $\mathbf{T}%
_{2;-h}$ transforms a solution $u$ of the equation $(\partial_{x}^{2}%
+\omega^{2})u=0$ into a solution of the equation $(A_{2}+\omega^{2}%
)v=(\partial_{x}^{2}-q_{2}(x)+\omega^{2})v=0$ with the correspondence of
initial conditions
\begin{equation}
\mathbf{T}_{2;-h}[u](0)=u(0),\quad(\mathbf{T}_{2;-h}[u])^{\prime}%
(0)=u^{\prime}(0)-hu(0). \label{T2IC}%
\end{equation}
Hence we have to verify the same properties for the operator $T_{2}$. Consider
a solution $u$ of the equation $(\partial_{x}^{2}+\omega^{2})u=0$ for some
particular $\omega\in\mathbb{C}$. The function $u^{\prime}$ is again a
solution of this equation, and by Proposition \ref{ThMapsSolutions},
$\mathbf{T}_{1}u$ is a solution of the equation $(\partial_{x}^{2}%
-q_{1}+\omega^{2})v=0$, so with the aid of the factorization \eqref{A1factor}
we obtain
\begin{equation}
\frac{1}{f}\partial_{x}f^{2}\partial_{x}\frac{1}{f}\mathbf{T}_{1}[u^{\prime
}]=-\omega^{2}\mathbf{T}_{1}[u^{\prime}]. \label{factoredT1du}%
\end{equation}
Let us apply the operator $A_{2}$ to the function $T_{2}u$ and use the
factorization \eqref{A2factor}.
\begin{equation}
A_{2}T_{2}[u](x)=\Big(f\partial_{x}\frac{1}{f^{2}}\partial_{x}f\Big)\left(
\frac{1}{f(x)}\bigg(\int_{0}^{x}f(s)\mathbf{T}_{1}[u^{\prime}%
](s)\,ds+u(0)\bigg)\right)  =\Big(f\partial_{x}\frac{1}{f}\Big)\big(\mathbf{T}%
_{1}[u^{\prime}]\big). \label{A2T2u}%
\end{equation}
Applying $\frac{1}{f}\partial_{x}f$ to both sides of \eqref{A2T2u} due to
\eqref{factoredT1du} we obtain
\[
\Big(\frac{1}{f}\partial_{x}f\Big)\big(A_{2}T_{2}[u]\big)=\Big(\frac{1}%
{f}\partial_{x}f^{2}\partial_{x}\frac{1}{f}\Big)\big(\mathbf{T}_{1}[u^{\prime
}]\big)=-\omega^{2}\mathbf{T}_{1}[u^{\prime}]=-\omega^{2}\Big(\frac{1}%
{f}\partial_{x}f\Big)\big(T_{2}[u]\big).
\]
Hence the function $A_{2}T_{2}[u]$ may differ from the function $-\omega
^{2}T_{2}[u]$ only by $c/f$, where $c$ is a constant. To find the value of the
constant we compute the values of both expressions for $x=0$. We have
\[
T_{2}[u](0)=u(0)
\]
and%
\[
A_{2}T_{2}[u](x)=\Big(f\partial_{x}\frac{1}{f}\Big)\big(\mathbf{T}%
_{1}[u^{\prime}](x)\big)=-\frac{f^{\prime}(x)}{f(x)}\mathbf{T}_{1}[u^{\prime
}](x)+\big(\mathbf{T}_{1}[u^{\prime}](x)\big)^{\prime}.
\]
As $u^{\prime}$ is a solution of $(\partial_{x}^{2}+\omega^{2})v=0$, by
Proposition \ref{ThMapsSolutions} we have $\mathbf{T}_{1}[u^{\prime
}](0)=u^{\prime}(0)$, $\big(\mathbf{T}_{1}[u^{\prime}]\big)^{\prime
}(0)=u^{\prime\prime}(0)+hu^{\prime}(0)=-\omega^{2}u(0)+hu^{\prime}(0)$ (the
last equality holds since the function $u$ is also a solution of
$(\partial_{x}^{2}+\omega^{2})v=0$). Therefore
\[
A_{2}T_{2}[u](0)=-hu^{\prime}(0)+\omega u(0)+hu^{\prime}(0)=-\omega^{2}u(0).
\]
Hence $c=0$ and the operator $T_{2}$ maps solutions of the equation
$(\partial_{x}^{2}+\omega^{2})u=0$ into solutions of the equation
$(A_{2}+\omega^{2})v=0$ for any $\omega\in\mathbb{C}$. To finish the proof, we
have to check conditions \eqref{T2IC} for the operator $T_{2}$. We have
\[
T_{2}[u](0)=\frac{u(0)}{f(0)}=u(0)
\]
and%
\[
\big(T_{2}[u]\big)^{\prime}(x)=-\frac{f^{\prime}(x)}{f(x)}T_{2}[u](x)+\frac
{1}{f(x)}\cdot f(x)\mathbf{T}_{1}[u^{\prime}](x)=-\frac{f^{\prime}(x)}%
{f(x)}T_{2}[u](x)+\mathbf{T}_{1}[u^{\prime}](x),
\]
hence for $x=0$ we obtain
\[
\big(T_{2}[u]\big)^{\prime}(0)=-hT_{2}[u](0)+\mathbf{T}_{1}[u^{\prime
}](0)=-hu(0)+u^{\prime}(0),
\]
which finishes the proof.
\end{proof}

Now we show that the operator $T_{2}$ can be written as a Volterra integral
operator and, as a consequence, extended by continuity to a wider class of functions.

\begin{theorem}
\label{Th T2Volterra} The operator $T_{2}$ admits a representation as the
Volterra integral operator
\begin{equation}
T_{2}[u](x)=u(x)+\int_{-x}^{x}\mathbf{K}_{2}(x,t;-h)u(t)\,dt, \label{T2}%
\end{equation}
with the kernel
\begin{equation}
\mathbf{K}_{2}(x,t;-h)=-\frac{1}{f(x)}\bigg(\int_{-t}^{x}\partial
_{t}\mathbf{K}_{1}(s,t;h)f(s)\,ds+\frac{h}{2}f(-t)\bigg). \label{K2}%
\end{equation}

\end{theorem}

\begin{proof}
Consider the expression
\[
\int_{0}^{x}f(t)\mathbf{T}_{1}[u^{\prime}](t)\,dt+u(0)=\int_{0}^{x}%
f(t)u^{\prime}(t)\,dt+\int_{0}^{x}f(t)\int_{-t}^{t}\mathbf{K}_{1}%
(t,s;h)u^{\prime}(s)\,ds\,dt+u(0).
\]
Suppose that $x>0$ (the opposite case is similar). We integrate by parts the
first integral and change the order of integration in the second integral,
\begin{multline}
\int_{0}^{x}f(t)\mathbf{T}_{1}[u^{\prime}](t)\,dt+u(0)=\label{T2Volterra1}\\
f(x)u(x)-f(0)u(0)-\int_{0}^{x}f^{\prime}(t)u(t)\,dt+\int_{-x}^{x}u^{\prime
}(s)\int_{|s|}^{x}\mathbf{K}_{1}(t,s;h)f(t)\,dt\,ds+u(0)=\\
f(x)u(x)-\int_{0}^{x}f^{\prime}(t)u(t)\,dt+u(s)\int_{|s|}^{x}\mathbf{K}%
_{1}(t,s;h)f(t)\,dt\bigg|_{-x}^{x}-\int_{-x}^{x}u(s)\frac{d}{ds}%
\bigg(\int_{|s|}^{x}\mathbf{K}_{1}(t,s;h)f(t)\,dt\bigg)ds=\\
f(x)u(x)-\int_{0}^{x}f^{\prime}(t)u(t)\,dt-\int_{-x}^{x}u(s)\frac{d}%
{ds}\bigg(\int_{|s|}^{x}\mathbf{K}_{1}(t,s;h)f(t)\,dt\bigg)ds.
\end{multline}

To continue the proof we will use the extension of the kernel $\mathbf{K}_{1}
$ onto the square $|x|\leq a,|t|\leq a$ as a continuously differentiable
function. Consider the integral
\[
\int_{-s}^{|s|}\mathbf{K}_{1}(t,s;h)f(t)\,dt.
\]
For $s\leq0$ it equals zero. For $s>0$ note that by Theorem
\ref{Th Transmutation of Powers}, $\mathbf{T}_{1}[1](x)=f(x)$, so we obtain
from \eqref{Tinverse}
\[
1=\mathbf{T}_{1}^{-1}[f](x)=f(x)-\int_{-x}^{x}\mathbf{K}_{1}(t,x;h)f(t)\,dt.
\]
Hence
\begin{equation}
\int_{-s}^{|s|}\mathbf{K}_{1}(t,s;h)f(t)\,dt=%
\begin{cases}
0, & \text{if\ }s\leq0,\\
f(s)-1, & \text{if\ }s>0.
\end{cases}
\label{1Tinv}%
\end{equation}
Combining \eqref{T2Volterra1} and \eqref{1Tinv}, we get
\begin{multline*}
\int_{0}^{x}f(t)\mathbf{T}_{1}[u^{\prime}](t)\,dt+u(0)=f(x)u(x)-\int_{0}%
^{x}f^{\prime}(t)u(t)\,dt-\\
-\int_{-x}^{x}u(s)\frac{d}{ds}\bigg(\int_{-s}^{x}\mathbf{K}_{1}%
(t,s;h)f(t)\,dt-\int_{-s}^{|s|}\mathbf{K}_{1}(t,s;h)f(t)\,dt\bigg)ds=\\
=f(x)u(x)-\int_{0}^{x}f^{\prime}(t)u(t)\,dt-\int_{-x}^{x}u(s)\frac{d}%
{ds}\bigg(\int_{-s}^{x}\mathbf{K}_{1}(t,s;h)f(t)\,dt\bigg)ds+\int_{0}%
^{x}u(s)f^{\prime}(s)\,ds=\\
=f(x)u(x)-\int_{-x}^{x}u(s)\bigg(\int_{-s}^{x}\frac{d}{ds}\mathbf{K}%
_{1}(t,s;h)f(t)\,dt+\mathbf{K}_{1}(-s,s;h)f(-s)\bigg)ds=\\
=f(x)u(x)-\int_{-x}^{x}u(s)\bigg(\int_{-s}^{x}\frac{d}{ds}\mathbf{K}%
_{1}(t,s;h)f(t)\,dt+\frac{h}{2}f(-s)\bigg)ds.
\end{multline*}
Hence
\[
T_{2}[u](x)=\frac{1}{f(x)}\bigg(\int_{0}^{x}f(\eta)\mathbf{T}_{1}[u^{\prime
}](\eta)\,d\eta+u(0)\bigg)=u(x)+\int_{-x}^{x}\mathbf{K}_{2}(x,t;-h)u(t)\,dt,
\]
where the kernel $\mathbf{K}_{2}$ is given by \eqref{K2}.
\end{proof}

\begin{remark}
\label{Remark1} Note that in the proof of Theorem \ref{Th T2Volterra} we made
use only of the fact that $u\in C^{1}[-a,a]$ and never required that $u$ be a
solution of the equation $\partial_{x}^{2}u+\omega^{2}u=0$. Therefore both
representations for the operator $T_{2}$ obtained in Theorems
\ref{Th T2Integral} and \ref{Th T2Volterra} coincide on any function $u\in
C^{1}[-a,a]$.
\end{remark}

\begin{remark}
If the integral kernel $\mathbf{K}_{1}$ is known only in the domain $|x|\leq
a,|t|\leq|x|$, from \eqref{T2Volterra1} it is also possible to obtain the
expression of the integral kernel for the Volterra integral operator
representing $T_{2}$. Since
\begin{multline*}
\frac{d}{dt}\bigg(\int_{|t|}^{x}\mathbf{K}_{1}(s,t;h)f(s)\,ds\bigg)=\int
_{|t|}^{x}\frac d{dt} \mathbf{K}_{1}(s,t;h)f(s)\,ds- \mathbf{K}_{1}%
(|t|,t;h)f(|t|)\cdot(|t|)^{\prime}=\\
=%
\begin{cases}
\int_{|t|}^{x}\frac d{dt} \mathbf{K}_{1}(s,t;h)f(s)\,ds-\frac h2f(t) -
\frac{f(t)}2\int_{0}^{t} q_{1}(s)\,ds & \text{for\ } t\ge0,\\
\int_{|t|}^{x}\frac d{dt} \mathbf{K}_{1}(s,t;h)f(s)\,ds+\frac h2f(-t) &
\text{for\ }t<0,
\end{cases}
\end{multline*}
the integral kernel for $x>0$ is given by the expression
\[
\mathbf{K}_{2}(x,t;-h)=%
\begin{cases}
-\frac{1}{f(x)}\Big(f^{\prime}(t)+\int_{|t|}^{x}\frac{d}{dt}\big(\mathbf{K}%
_{1}(s,t;h)\big)f(s)\,ds-\frac{h}{2}f(t)-\frac{f(t)}{2}\int_{0}^{t}%
q_{1}(s)\,ds\Big) & \text{if\ }t\geq0,\\
-\frac{1}{f(x)}\Big(\int_{|t|}^{x}\frac{d}{dt}\big(\mathbf{K}_{1}%
(s,t;h)\big)f(s)\,ds+\frac{h}{2}f(-t)\Big) & \text{if\ }t<0.
\end{cases}
\]

\end{remark}

\begin{corollary}
The operator $T_{2}$ given by \eqref{T2} with the kernel \eqref{K2} coincides
with $\mathbf{T}_{2}$ on $C[-a,a]$.
\end{corollary}

\begin{proof}
By Theorems \ref{Th T2Integral} and \ref{Th T2Volterra} the Volterra operators
$T_{2}$ and $\mathbf{T}_{2}$ coincide on the set of finite linear combinations
of solutions of the equations $(\partial_{x}^{2}+\omega^{2})u=0,\ \omega
\in\mathbb{C}$. Since this set is dense in $C[-a,a]$, by continuity of $T_{2}$
and $\mathbf{T}_{2}$ we obtain that they coincide on the whole $C[-a,a]$.
\end{proof}

The next corollary follows immediately from Remark \ref{Remark1}.

\begin{corollary}
The operator $T_{2}$ given by \eqref{T2sol} coincides with $\mathbf{T}_{2}$ on
$C^{1}[-a,a]$.
\end{corollary}

Operator $A_{1}$ is the Darboux transformation of the operator $A_{2}$ with
respect to the solution $1/f$, hence we obtain another relation between the
operators $\mathbf{T}_{1}$ and $\mathbf{T}_{2}$.

\begin{corollary}
For any function $u\in C^{1}[-a,a]$ the equality
\begin{equation}
\mathbf{T}_{1}[u](x)=f(x)\bigg(\int_{0}^{x}\frac{1}{f(\eta)}\mathbf{T}%
_{2}[u^{\prime}](\eta)\,d\eta+u(0)\bigg) \label{T1sol}%
\end{equation}
is valid.
\end{corollary}

From the second commutative diagram at the beginning of this subsection we may
deduce some commutation relations between the operators $\mathbf{T}_{1}$,
$\mathbf{T}_{2}$ and $d/dx$. The proof immediately follows from \eqref{T2sol}
and \eqref{T1sol}.

\begin{corollary}
\label{Cor Commutation Relations}The following operator equalities hold on
$C^{1}[-a,a]$:
\begin{align}
\partial_{x}f\mathbf{T}_{2}  &  =f\mathbf{T}_{1}\partial_{x}
\label{CommutT1dx}\\
\partial_{x}\frac{1}{f}\mathbf{T}_{1}  &  =\frac{1}{f}\mathbf{T}_{2}%
\partial_{x}. \label{CommutT2dx}%
\end{align}

\end{corollary}

Commutation relations \eqref{CommutT1dx} and \eqref{CommutT2dx} allow us to
prove a more general version of \cite[Theorem 11]{CKT} which is also Theorem
\ref{Th Transmutation} with different conditions on the potential $q$.

\begin{theorem}
\label{Th Transmutation C} Under the conditions of Theorem
\ref{Th Transmutation of Powers} the operator \eqref{Tmain} with the kernel
defined by \eqref{Kmain} satisfies
\[
\left(  -\frac{d^{2}}{dx^{2}}+q(x)\right)  \mathbf{T}_{h}[u]=\mathbf{T}%
_{h}\left[  -\frac{d^{2}u}{dx^{2}}\right]
\]
for any $u\in C^{2}[-a,a]$.
\end{theorem}

\begin{proof}
Let $\mathbf{T}_{2}$ be the transmutation operator of the Darboux
transformation of the operator $A=-\frac{d^{2}}{dx^{2}}+q(x)$ with respect to
the solution $f$. Let $u\in C^{2}[-a,a]$, then $u^{\prime}\in C^{1}[-a,a]$ and
we obtain by \eqref{CommutT1dx}, \eqref{CommutT2dx} and \eqref{A1factor}
\[
\mathbf{T}_{h}\left[  \frac{d^{2}u}{dx^{2}}\right]  =\frac{1}{f}\frac{d}%
{dx}\big[f\mathbf{T}_{2}u^{\prime}\big]=\frac{1}{f}\frac{d}{dx}f^{2}\frac
{d}{dx}\frac{1}{f}\mathbf{T}_{h}u=\left(  \frac{d^{2}}{dx^{2}}-q(x)\right)
\mathbf{T}_{h}u.
\]

\end{proof}

In \cite{KMoT} the following notion of generalized derivatives was introduced.
Consider a function $g$ assuming that both $f$ and $g$ possess the derivatives
of all orders up to the order $n$ on the segment $[-a,a]$. Then in $[-a,a]$
the following generalized derivatives are defined
\begin{align*}
\gamma_{0}(g)(x) &  =g(x),\\
\gamma_{k}(g)(x) &  =\big(f^{2}(x)\big)^{(-1)^{k-1}}\big(\gamma_{k-1}%
(g)\big)^{\prime}(x)
\end{align*}
for $k=1,2,\ldots,n$.

Let a function $u$ be defined by the equality
\[
g=\frac{1}{f}\mathbf{T}_{1}u,
\]
and assume that $u\in C^{n}[-a,a]$. Note that below we do not necessarily
require that the functions $f$ and $g$ be from $C^{n}[-a,a]$. With the use of
\eqref{CommutT1dx} and \eqref{CommutT2dx} we have
\begin{align*}
\gamma_{1}(g) &  =f^{2}\cdot\Big(\frac{1}{f}\mathbf{T}_{1}u\Big)^{\prime
}=f^{2}\cdot\frac{1}{f}\mathbf{T}_{2}u^{\prime}=f\mathbf{T}_{2}u^{\prime},\\
\gamma_{2}(g) &  =\frac{1}{f^{2}}\cdot\Big(f\mathbf{T}_{2}u^{\prime
}\Big)^{\prime}=\frac{1}{f^{2}}\cdot f\mathbf{T}_{1}u^{\prime\prime}=\frac
{1}{f}\mathbf{T}_{1}u^{\prime\prime}.
\end{align*}
By induction we obtain the following corollary.

\begin{corollary}
\label{GeneralDerivTransm} Let $u\in C^{n}[-a,a]$ and $g=\frac{1}{f}%
\mathbf{T}_{1}u$. Then
\[
\gamma_{k}(g)=f\mathbf{T}_{2}u^{(k)}\qquad\text{if\ }k\text{\ is odd,}\ k\leq
n,
\]
and
\[
\gamma_{k}(g)=\frac{1}{f}\mathbf{T}_{1}u^{(k)}\qquad\text{if\ }k\text{\ is
even,}\ k\leq n.
\]

\end{corollary}

To finish this subsection, let us consider how the operator $\mathbf{T}%
_{2;-h}$ acts on powers of the independent variable. By Theorem
\ref{Th Transmutation of Powers} we have to construct the system of functions
$\{\varphi_{k}\}$ by formulas \eqref{X1}--\eqref{X3} and \eqref{phik} starting
with the function $1/f$. As it can be seen from \eqref{X1}--\eqref{X3}, for
any nonnegative integer $k$ we have
\[
\widetilde{X}_{1/f}^{(k)}=X_{f}^{(k)}\qquad\text{and}\qquad X_{1/f}%
^{(k)}=\widetilde{X}_{f}^{(k)},
\]
where the subindex $f$ or $1/f$ corresponds to the starting function used in
\eqref{X1}--\eqref{X3}. Therefore constructing the functions $X^{(k)}$,
$\widetilde{X}^{(k)}$ according to \eqref{X1}--\eqref{X3} and defining
\begin{equation}
\psi_{k}(x)=%
\begin{cases}
\frac{1}{f(x)}X^{(k)}, & k\ \text{even},\\
\frac{1}{f(x)}\widetilde{X}^{(k)}, & k\ \text{odd},
\end{cases}
\label{psik}%
\end{equation}
(here the \textquotedblleft second half\textquotedblright\ of the formal
powers \eqref{X1}--\eqref{X3} is used, cf. (\ref{phik})) we obtain the
following statement.

\begin{proposition}
Let a potential $q_{1}$ and a function $f$ be as in Theorem
\ref{Th Transmutation of Powers}. Let the operator $\mathbf{T}_{2;-h}$ be the
transmutation operator for the Darboux transformed operator $A_{2}$. Then
$\mathbf{T}_{2;-h}$ transforms $x^{k}$ into $\psi_{k}(x)$ for any
$k\in\mathbb{N}_{0}$.
\end{proposition}

\subsection{Examples}

\label{subsectExamples} We start with the operator $A_{0}=d^{2}/dx^{2}$. We
have to pick up such a solution $f$ of the equation $A_{0}f=0$ that
$f^{\prime}/f\neq0$. This is in order to obtain an operator $A_{1}\neq A_{0}$
as a result of the Darboux transformation of $A_{0}$. As such solution
consider, e.g., $f_{0}(x)=x+1$. Both $f_{0}$ and $1/f_{0}$ are bounded on any
segment $[-a,a]\subset(-1;1)$ and the Darboux transformed operator has the
form $A_{1}=\frac{d^{2}}{dx^{2}}-\frac{2}{(x+1)^{2}}$.

The transmutation operator $T$ for $A_{0}$ is obviously an
identity operator and $K_{0}(x,t;0)=0$. Since $f_{0}^{\prime}(0)=1$, we look
for the parametrized operator $\mathbf{T}_{0;1}$. Its kernel is given by
\eqref{KmainChangeOfH}: $\mathbf{K}_{0}(x,t;1)=1/2$. From Theorem
\ref{Th T2Volterra} we obtain the transmutation kernel for the operator
$A_{1}$
\begin{equation}
\mathbf{K}_{1}(x,t;-1)=-\frac{1}{x+1}\cdot\frac{1-t}{2}=\frac{t-1}{2(x+1)},
\label{ExampleK1}%
\end{equation}
the kernel from Example \ref{modelExample}.

To obtain a less trivial example consider again the operator $A_{1}%
=\frac{d^{2}}{dx^{2}}-\frac{2}{(x+1)^{2}}$ and the function $f_{1}%
(x)=(x+1)^{2}$ as a solution of $A_{1}f=0$. Since $h=f_{1}^{\prime}(0)=2$, we
compute $\mathbf{K}_{1}(x,t;2)$ from \eqref{ExampleK1} using
\eqref{KmainChangeOfH}
\[
\mathbf{K}_{1}(x,t;2)=\frac{3x^{2}+6x+4-3t^{2}+2t}{4(x+1)}.
\]
The Darboux transformation of the operator $A_{1}$ with respect to the
solution $f_{1}$ is the operator $A_{2}=\frac{d^{2}}{dx^{2}}-\frac
{6}{(x+1)^{2}}$ and by Theorem \ref{Th T2Volterra} the transmutation operator
$\mathbf{T}_{2;-2}$ for $A_{2}$ is given by the Volterra integral operator
\eqref{Tmain} with the kernel
\[
\mathbf{K}_{2}(x,t;-2)=-\frac{1}{(x+1)^{2}}\bigg(\int_{-t}^{x}\frac
{-3t+1}{2(s+1)}(s+1)^{2}\,ds+(1-t)^{2}\bigg)=\frac{(3t-1)(x+1)^{2}%
-3(t-1)^{2}(t+1)}{4(x+1)^{2}}.
\]

This procedure may be continued iteratively. Consider operators
\[
A_{n}:=\frac{d^{2}}{dx^{2}}-\frac{n(n+1)}{(x+1)^{2}}.
\]
$f_{n}(x)=(x+1)^{n+1}$ is a solution of the equation $A_{n}f=0$. The Darboux
transformation of the operator $A_{n}$ with respect to the solution $f_{n}$ is
the operator
\[
\frac{d^{2}}{dx^{2}}-2\Big(\frac{f_{n}^{\prime}(x)}{f_{n}(x)}\Big)^{2}%
+\frac{n(n+1)}{(x+1)^{2}}=\frac{d^{2}}{dx^{2}}-\frac{(n+1)(n+2)}{(x+1)^{2}},
\]
i.e., exactly the operator $A_{n+1}$. If we know $\mathbf{K}_{n}(x,t;-n)$ for
the operator $A_{n}$, by \eqref{KmainChangeOfH} we compute the kernel
$\mathbf{K}_{n}(x,t;n+1)$ corresponding to the solution $f_{n}(x)$ and by
Theorem \ref{Th T2Volterra} we may calculate the kernel $\mathbf{K}%
_{n+1}(x,t;-n-1)$. Careful analysis shows that we have to integrate only
polynomials in all integrals involved, so the described procedure can be
performed up to any fixed $n$.

Consider the Schr\"{o}dinger equation
\begin{equation}
u^{\prime\prime}+2\sech^{2}(x)\,u=u. \label{ExampleSoliton}%
\end{equation}
This equation appears in soliton theory and as an example of a reflectionless
potential in the one-dimensional quantum scattering theory (see, e.g.
\cite{Lamb}). Equation \eqref{ExampleSoliton} can be obtained as a result of
the Darboux transformation of the equation $u^{\prime\prime}=u$ with respect
to the solution $f(x)=\cosh x$. The transmutation operator for the operator
$A_{1}=\partial_{x}^{2}-1$ was calculated in \cite[Example 3]{CKT}. Its kernel
is given by the expression
\[
\mathbf{K}_{1}(x,t;0)=-\frac{1}{2}\frac{\sqrt{x^{2}-t^{2}}I_{1}(\sqrt
{x^{2}-t^{2}})}{x-t},
\]
where $I_{1}$ is the modified Bessel function of the first kind. Hence from
Theorem \ref{Th T2Volterra} we obtain the transmutation kernel for the
operator $A_{2}=\partial_{x}^{2}+2\sech^{2}x-1$
\[
\mathbf{K}_{2}(x,t;0)=\frac{1}{2\cosh(x)}\int_{-t}^{x}\left(  \frac
{I_{0}(\sqrt{s^{2}-t^{2}})t}{s-t}+\frac{\sqrt{s^{2}-t^{2}}I_{1}(\sqrt
{s^{2}-t^{2}})}{(s-t)^{2}}\right)  \cosh s\,ds.
\]

\section{Transmutation operator for the one-dimensional Dirac equation with a
Lorentz scalar potential\label{Sect Transmutation Dirac}}

One-dimensional Dirac equations with Lorentz scalar potentials are widely
studied (see, for example, \cite{Casahorran, Chen, Hiller, Ho, JP, KP, KhR,
NT, RV, Rukeng} and \cite{NPS} for intertwining techniques for it).

According to \cite{NT} the Dirac equation in one space dimension with a
Lorentz scalar potential can be written as
\begin{align}
(\partial_{x}+m+S(x))\psi_{1}  &  =E\psi_{2},\label{Dirac1}\\
(-\partial_{x}+m+S(x))\psi_{2}  &  =E\psi_{1}, \label{Dirac2}%
\end{align}
where $m$ ($m>0$) is the mass and $S(x)$ is a Lorentz scalar. Denote
$\eta=m+S$ and write the system \eqref{Dirac1}, \eqref{Dirac2} in a matrix
form as
\begin{equation}%
\begin{pmatrix}
\partial_{x}+\eta & 0\\
0 & \partial_{x}-\eta
\end{pmatrix}%
\begin{pmatrix}
\psi_{1}\\
\psi_{2}%
\end{pmatrix}
=E%
\begin{pmatrix}
0 & 1\\
-1 & 0
\end{pmatrix}%
\begin{pmatrix}
\psi_{1}\\
\psi_{2}%
\end{pmatrix}
. \label{DiracMatrix}%
\end{equation}
In order to apply the results on the transmutation operators and
factorizations \eqref{A1factor}, \eqref{A2factor} we consider a function $f$
such that
\[
\frac{f^{\prime}(x)}{f(x)}=-\eta=-m-S(x).
\]
We can take $f(x)=\exp\left(  -\int_{0}^{x}(m+S(s))\,ds\right)  $, then
$f(0)=1 $ and $f$ does not vanish. Suppose the operators $\mathbf{T}_{1}$ and
$\mathbf{T}_{2}$ are transmutations for the operators $A_{1}=\big(\partial
_{x}+\frac{f^{\prime}}{f}\big)\big(\partial_{x}-\frac{f^{\prime}}{f}\big)$ and
$A_{2}=\big(\partial_{x}-\frac{f^{\prime}}{f}\big)\big(\partial_{x}%
+\frac{f^{\prime}}{f}\big)$ respectively (corresponding to functions $f$ and
$1/f$ in the sense of Proposition \ref{ThMapsSolutions}). We will look for a
solution of equation \eqref{DiracMatrix} in the form
\[%
\begin{pmatrix}
\psi_{1}\\
\psi_{2}%
\end{pmatrix}
=%
\begin{pmatrix}
\mathbf{T}_{1} & 0\\
0 & \mathbf{T}_{2}%
\end{pmatrix}%
\begin{pmatrix}
u_{1}\\
u_{2}%
\end{pmatrix}
,
\]
where $u_{1}$ and $u_{2}$ are some functions. From the commutation relations
\eqref{CommutT1dx} and \eqref{CommutT2dx} we have
\[%
\begin{pmatrix}
\partial_{x}-\frac{f^{\prime}}{f} & 0\\
0 & \partial_{x}+\frac{f^{\prime}}{f}%
\end{pmatrix}%
\begin{pmatrix}
\mathbf{T}_{1} & 0\\
0 & \mathbf{T}_{2}%
\end{pmatrix}%
\begin{pmatrix}
u_{1}\\
u_{2}%
\end{pmatrix}
=%
\begin{pmatrix}
f\partial_{x}\frac{1}{f}\mathbf{T}_{1} & 0\\
0 & \frac{1}{f}\partial_{x}f\mathbf{T}_{2}%
\end{pmatrix}%
\begin{pmatrix}
u_{1}\\
u_{2}%
\end{pmatrix}
=%
\begin{pmatrix}
\mathbf{T}_{2} & 0\\
0 & \mathbf{T}_{1}%
\end{pmatrix}%
\begin{pmatrix}
\partial_{x} & 0\\
0 & \partial_{x}%
\end{pmatrix}%
\begin{pmatrix}
u_{1}\\
u_{2}%
\end{pmatrix}
,
\]
hence
\[%
\begin{pmatrix}
\mathbf{T}_{2} & 0\\
0 & \mathbf{T}_{1}%
\end{pmatrix}%
\begin{pmatrix}
\partial_{x} & 0\\
0 & \partial_{x}%
\end{pmatrix}%
\begin{pmatrix}
u_{1}\\
u_{2}%
\end{pmatrix}
=E%
\begin{pmatrix}
0 & 1\\
-1 & 0
\end{pmatrix}%
\begin{pmatrix}
\mathbf{T}_{1} & 0\\
0 & \mathbf{T}_{2}%
\end{pmatrix}%
\begin{pmatrix}
u_{1}\\
u_{2}%
\end{pmatrix}
.
\]
By multiplying both sides by the inverse matrix $\Bigl(%
\begin{smallmatrix}
\mathbf{T}_{2}^{-1} & 0\\
0 & \mathbf{T}_{1}^{-1}%
\end{smallmatrix}
\Bigr)$ we obtain
\begin{multline*}%
\begin{pmatrix}
\partial_{x} & 0\\
0 & \partial_{x}%
\end{pmatrix}%
\begin{pmatrix}
u_{1}\\
u_{2}%
\end{pmatrix}
=E%
\begin{pmatrix}
\mathbf{T}_{2}^{-1} & 0\\
0 & \mathbf{T}_{1}^{-1}%
\end{pmatrix}%
\begin{pmatrix}
0 & 1\\
-1 & 0
\end{pmatrix}%
\begin{pmatrix}
\mathbf{T}_{1} & 0\\
0 & \mathbf{T}_{2}%
\end{pmatrix}%
\begin{pmatrix}
u_{1}\\
u_{2}%
\end{pmatrix}
=\\
=E%
\begin{pmatrix}
0 & \mathbf{T}_{2}^{-1}\\
-\mathbf{T}_{1}^{-1} & 0
\end{pmatrix}%
\begin{pmatrix}
\mathbf{T}_{1} & 0\\
0 & \mathbf{T}_{2}%
\end{pmatrix}%
\begin{pmatrix}
u_{1}\\
u_{2}%
\end{pmatrix}
=E%
\begin{pmatrix}
0 & 1\\
-1 & 0
\end{pmatrix}%
\begin{pmatrix}
u_{1}\\
u_{2}%
\end{pmatrix}
.
\end{multline*}
Therefore the operator $\bigl(%
\begin{smallmatrix}
\mathbf{T}_{1} & 0\\
0 & \mathbf{T}_{2}%
\end{smallmatrix}
\bigr)$ transmutes any solution $\bigl(%
\begin{smallmatrix}
u_{1}\\
u_{2}%
\end{smallmatrix}
\bigr)$ of the system
\begin{align}
u_{1}^{\prime}  &  =Eu_{2}\label{DiracTriv1}\\
u_{2}^{\prime}  &  =-Eu_{1} \label{DiracTriv2}%
\end{align}
into the solution $\bigl(%
\begin{smallmatrix}
\psi_{1}\\
\psi_{2}%
\end{smallmatrix}
\bigr)$ of the system \eqref{Dirac1}, \eqref{Dirac2} with the initial
conditions $\psi_{1}(0)=u_{1}(0)$, $\psi_{2}(0)=u_{2}(0)$. And vice versa if
$\bigl(%
\begin{smallmatrix}
\psi_{1}\\
\psi_{2}%
\end{smallmatrix}
\bigr)$ is a solution of the system \eqref{Dirac1}, \eqref{Dirac2}, then the
operator $\Bigl(%
\begin{smallmatrix}
\mathbf{T}_{1}^{-1} & 0\\
0 & \mathbf{T}_{2}^{-1}%
\end{smallmatrix}
\Bigr)$ transmutes it into the solution $\bigl(%
\begin{smallmatrix}
u_{1}\\
u_{2}%
\end{smallmatrix}
\bigr)$ of \eqref{DiracTriv1}, \eqref{DiracTriv2} such that $u_{1}(0)=\psi
_{1}(0)$, $u_{2}(0)=\psi_{2}(0)$.

\section{Conclusions\label{Sect Conclusions}}

An explicit representation for the transmutation operator corresponding to a
Darboux transformed Schr\"{o}dinger operator is given in terms of the
transmutation kernel for its superpartner and as a corollary the transmutation
operator for the one-dimensional Dirac system with a scalar potential is
obtained. Several examples of explicitly constructed transmutation operators
are given. We expect that the techniques developed in the present paper will
be used in practical applications of the transmutation operators.


{\footnotesize
\begin{thebibliography}{99}                                                                                               %
\itemsep=0pt 

\bibitem {BS}{\footnotesize V. G. Bagrov and B. F. Samsonov. Darboux
transformation, factorization, and supersymmetry in one-dimensional quantum
mechanics, Teoret. Mat. Fiz. 1995, vol. 104, no. 2, 356--367 (in Russian);
translation in Theoret. and Math. Phys. 1995, vol. 104, no. 2, 1051-1060. }

\bibitem {Gilbert}{\footnotesize H. Begehr and R. Gilbert. Transformations,
transmutations and kernel functions, vol. 1--2. Longman Scientific \&
Technical, Harlow, 1992.  }

\bibitem {CKT}{\footnotesize H. Campos, V. V. Kravchenko and S. Torba.
Transmutations, L-bases and complete families of solutions of the stationary
Schr\"{o}dinger equation in the plane. Submitted to Journal of Mathematical
Analysis and Applications. Available at arXiv:1109.5933  }

\bibitem {Carroll}{\footnotesize R. W. Carroll. Transmutation theory and
applications, Mathematics Studies, Vol. 117, North-Holland, 1985.  }

\bibitem {Casahorran}{\footnotesize J. Casahorr\'{a}n. Solving smultaneously
Dirac and Ricatti equations, Journal of Nonlinear Mathematical Physics, 1985,
v.5, n.4, 371-382.  }

\bibitem {Chen}{\footnotesize C.-Y. Chen. Exact solutions of the Dirac
equation with scalar and vector Hartmann potentials, Physics Letters A. 2005,
vol. 339, 283-287.  }

\bibitem {Cies}{\footnotesize J. L. Cie\'{s}li\'{n}ski. Algebraic construction
of the Darboux matrix revisited, J. Phys. A: Math. Theor. 2009, vol. 42,
404003. }

\bibitem {Fage}{\footnotesize M. K. Fage and N. I. Nagnibida. The problem of
equivalence of ordinary linear differential operators. Novosibirsk: Nauka,
1987 (in Russian).  }

\bibitem {Garab}{\footnotesize P. R. Garabedian. Partial differential
equations. New York--London: John Willey and Sons, 1964.  }

\bibitem {GHZh}{\footnotesize C. Gu, H. Hu, and Z. Zhou. Darboux
Transformations in Integrable Systems, Springer-Verlag, Berlin, 2005. }

\bibitem {HV}{\footnotesize A. D. Hemery and A. P. Veselov. Whittaker-Hill
equation and semifinite-gap Schr\"{o}dinger operators, J. Math. Phys., 2010,
Vol. 51, 072108; doi:10.1063/1.3455367. }

\bibitem {Hiller}{\footnotesize J. R. Hiller. Solution of the one-dimensional
Dirac equation with a linear scalar potential, Am. J. Phys., 2002, vol. 70
(5), 522-524.  }

\bibitem {Ho}{\footnotesize C.-L. Ho. Quasi-exact solvability of Dirac
equation with Lorentz scalar potential, Ann. Physics 2006, vol. 321, no. 9,
2170-2182.  }

\bibitem {JP}{\footnotesize R. Jackiw and S.-Y. Pi. Persistence of zero modes
in a gauged Dirac model for bilayer graphene, Phys. Rev. B, 2008, vol. 78,
132104. }

\bibitem {KP}{\footnotesize N. Kevlishvili, G. Piranishvili. Klein paradox in
modified Dirac and Salpeter equations, Fizika, 2003, vol. 9, no. 3,4, 57-61.
}

\bibitem {KhR}{\footnotesize K. V. Khmelnytskaya and H. C. Rosu. An
amplitude-phase (Ermakov--Lewis) approach for the Jackiw--Pi model of bilayer
graphene, J. Phys. A: Math. Theor., 2009, vol. 42, 042004 }

\bibitem {KrCV08}{\footnotesize V. V. Kravchenko. A representation for
solutions of the Sturm-Liouville equation. Complex Variables and Elliptic
Equations, 2008, v. 53, 775-789.  }

\bibitem {APFT}{\footnotesize V. V. Kravchenko. Applied pseudoanalytic
function theory. Basel: Birkh\"{a}user, Series: Frontiers in Mathematics,
2009.  }

\bibitem {KrCMA2011}{\footnotesize V. V. Kravchenko. On the completeness of
systems of recursive integrals. Communications in Mathematical Analysis, Conf.
03 2011 172--176.  }

\bibitem {KMoT}{\footnotesize V. V. Kravchenko, S. Morelos and S. Tremblay.
Complete systems of recursive integrals and Taylor series for solutions of
Sturm-Liouville equations. To appear in Mathematical Methods in the Applied
Sciences. }

\bibitem {KrPorter2010}{\footnotesize V. V. Kravchenko and R. M. Porter.
Spectral parameter power series for Sturm-Liouville problems. Mathematical
Methods in the Applied Sciences 2010, v. 33, 459-468.  }

\bibitem {Lamb}{\footnotesize G. L. Lamb. Elements of soliton theory. John
Wiley \& Sons, New York, 1980. }

\bibitem {LevitanInverse}{\footnotesize B. M. Levitan. Inverse Sturm-Liouville
problems. VSP, Zeist, 1987.  }

\bibitem {Lions57}{\footnotesize J. L. Lions. Solutions \'{e}l\'{e}mentaires
de certains op\'{e}rateurs diff\'{e}rentiels \`{a} coefficients variables.
Journ. de Math., t. 36, Fasc 1, 1957, 57-64.  }

\bibitem {Marchenko}{\footnotesize V. A. Marchenko. Sturm-Liouville operators
and applications. Basel: Birkh\"{a}user, 1986.  }

\bibitem {Matveev}{\footnotesize V. Matveev and M. Salle. Darboux
transformations and solitons. N.Y. Springer, 1991.  }

\bibitem {NT}{\footnotesize Y. Nogami, F. M. Toyama. Supersymmetry aspects of
the Dirac equation in one dimension with a Lorentz scalar potential, Physical
Review A., 1993, vol. 47, no. 3, 1708-1714.  }

\bibitem {NPS}{\footnotesize L. M. Nieto, A. A. Pecheritsin, and B. F.
Samsonov. Intertwining technique for the one-dimensional stationary Dirac
equation, Annals of Physics 2003, vol. 305, 151-189. }

\bibitem {PPS}{\footnotesize A. A. Pecheritsin, A. M. Pupasov and B. F.
Samsonov. Singular matrix Darboux transformations in the inverse-scattering
method, J. Phys. A: Math. Theor. 2011, vol. 44, 205305. }

\bibitem {RS}{\footnotesize C. Rogers, W. K. Schief. Backlund and Darboux
transformations: geometry and modern applications in soliton theory, Cambridge
University Press, 2002. }

\bibitem {Rosu}{\footnotesize H. Rosu. Short survey of Darboux
transformations, Proceedings of ``Symmetries in Quantum Mechanics and Quantum
Optics'', Burgos, Spain, 1999, 301-315.  }

\bibitem {RV}{\footnotesize R. K. Roychoudhory, Y. P. Varshni. Shifted 1/N
expansion and scalar potential in the Dirac equation, J. Phys. A: Math. Gen.
1987, vol. 20, L1083-L1087.  }

\bibitem {Sitnik}{\footnotesize S. M. Sitnik. Transmutations and applications:
a survey. arXiv:1012.3741v1 [math.CA], originally published in the book:
\textquotedblleft Advances in Modern Analysis and Mathematical
Modeling\textquotedblright\ Editors: Yu.F.Korobeinik, A.G.Kusraev,
Vladikavkaz: Vladikavkaz Scientific Center of the Russian Academy of Sciences
and Republic of North Ossetia--Alania, 2008, 226--293.  }

\bibitem {Rukeng}{\footnotesize R. Su, Yu Zhong and S. Hu. Solutions of Dirac
equation with one-dimensional scalarlike potential, Chinese Phys.Lett., 1991,
v.8, no.3, 114-117.  }

\bibitem {Trimeche}{\footnotesize K. Trimeche. Transmutation operators and
mean-periodic functions associated with differential operators. London:
Harwood Academic Publishers, 1988.  }

\bibitem {Vladimirov}{\footnotesize V. S. Vladimirov. Equations of
mathematical physics. Moskva: Nauka, 1981 (in Russian), English translation in
V.S. Vladimirov. Equations of mathematical physics (2nd English ed.), Moscow:
Mir Publishers, 1983. }

{\footnotesize
}
\end{thebibliography}
}
\end{document}